\documentclass{article}
\usepackage{graphicx}
\usepackage{amsmath,amsthm,amsfonts}
\usepackage{lineno}
\usepackage{color}

\newtheorem{theorem}{Theorem}[section]

\newtheorem{lemma}[theorem]{Lemma}

\newcommand{\conv}{\operatorname{Conv}}

\newcommand\blfootnote[1]{%
  \begingroup
  \renewcommand\thefootnote{}\footnote{#1}%
  \addtocounter{footnote}{-1}%
  \endgroup
}

\newcommand{\MyQuote}[1]{\vspace{0.5cm}\addtocounter{equation}{1}%
     \parbox{0.8\textwidth}{#1}\hspace*{1.5cm}(\arabic{equation})\\[0.5cm]}

\begin{document}

\newpage

\title{Carath\'eodory's Theorem in Depth}

\author{Ruy~Fabila-Monroy\thanks{Departamento de Matem\'aticas, Cinvestav, D.F., Mexico. \texttt{ruyfabila@math.cinvestav.edu.mx} Partially supported
by Conacyt of Mexico grant \textbf{253261}. }
	\and Clemens Huemer\thanks{Departament de Matem\'atiques, UPC, Barcelona, Spain.  \texttt{clemens.huemer@upc.edu} Partially 
	supported by projects \textbf{MTM2015-63791-R} and \textbf{Gen. Cat. DGR 2014SGR46}}}
\maketitle

\blfootnote{\begin{minipage}[l]{0.3\textwidth} \includegraphics[trim=10cm 6cm 10cm 5cm,clip,scale=0.15]{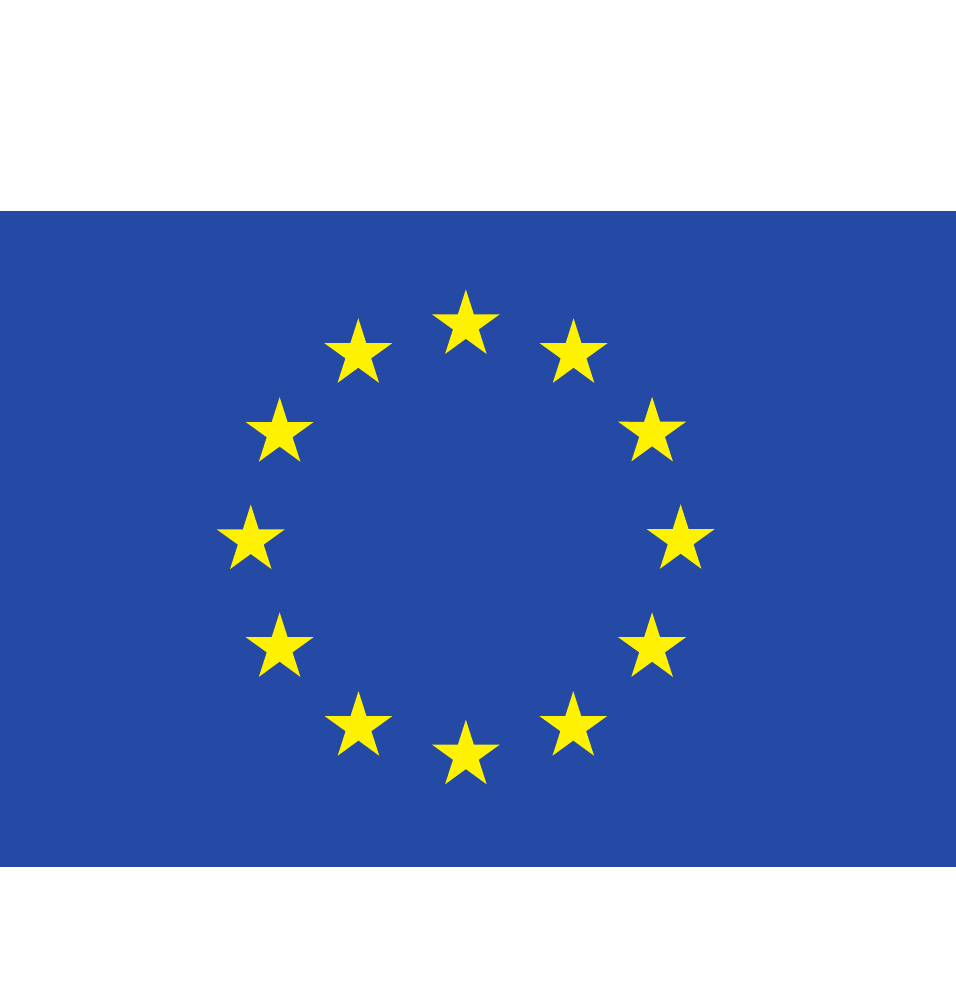} \end{minipage}  \hspace{-2cm} \begin{minipage}[l][1cm]{0.82\textwidth}
 	  This project has received funding from the European Union's Horizon 2020 research and innovation programme under the Marie Sk\l{}odowska-Curie grant agreement No 734922.
 	\end{minipage}}

\begin{abstract}

Let $X$ be a finite set of points in $\mathbb{R}^d$. The Tukey depth
of a point $q$ with respect to $X$ is the minimum number $\tau_X(q)$ of points of $X$ in a halfspace
containing $q$. In this paper we prove a depth version of Carath\'eodory's theorem.
In particular, we prove that there exist a constant $c$ (that depends
only on $d$ and $\tau_X(q)$) and pairwise disjoint sets $X_1,\dots, X_{d+1} \subset X$ such that the following holds. 
Each $X_i$ has at least $c|X|$ points, and for every choice of points
$x_i$ in $X_i$, $q$ is a convex combination of  
$x_1,\dots, x_{d+1}$. We also prove depth versions of Helly's and Kirchberger's theorems.
\end{abstract}

\section{Introduction}

Carath\'eodory's theorem was proven by Carath\'edory in 1907; it is one of the fundamental results in convex geometry.
For sets of points in $\mathbb{R}^d$, it states the following.
\begin{theorem}[\textbf{Carath\'edory's theorem~\cite{caratheodory}}]
Let $X$ be a set of points in $\mathbb{R}^d$. If a point $q$ is contained in the convex hull, $\conv(X)$, of $X$,
then there exist $x_1,\dots,x_{d+1} \in X$ such that $q \in \conv(\{x_1,\dots,x_{d+1}\})$.
\end{theorem}

The Tukey depth of a point $q$ with respect to a finite point set $X$ is a parameter $0 \le \tau_X(q) \le 1$ that measures how ``deep'' $q$
is inside of $\conv(X)$.  In this paper we prove the following depth dependent version of Carath\'eodory's theorem.
\begin{theorem}[\textbf{Depth Carath\'eodory's theorem}]\label{thm:car_depth}
Let $X$ be a finite set of points in $\mathbb{R}^d$ (or a Borel probability measure
 on $\mathbb{R}^d$) and let $q$ be a point in $\mathbb{R}^d$ of positive Tukey depth
 with respect to $X$. Then there exists a positive
 constant $c=c(d,\tau_X(q))$ (that depends only on $d$ and $\tau_X(q)$), such that the following holds. 
 There exist pairwise disjoint subsets $X_1,\dots, X_{d+1}$  of $X$ 
 such that for every choice of points $x_i \in X_i$, $q \in \conv(\{x_1,\dots,x_{d+1}\})$.
 Moreover, each of the $X_i$ consists of at least $c|X|$ points.
\end{theorem}
Informally, Theorem~\ref{thm:car_depth} states that the deeper $q$
is inside $\conv(X)$, the larger
subsets $X_1,\dots,X_{d+1}$ of $X$ exist, such that for every choice of points
$x_i \in X_i$, $q$ is contained in the convex hull  of
$\{x_1,\dots, x_{d+1}\}$. In this sense, Carath\'eodory's theorem states that sets $X_i$ of cardinality
one always exist, whenever $q$ has positive depth.

This paper is organized as follows. In Section~\ref{sec:helly} we review two other fundamental theorems in convex geometry: Helly's and Kirchberger's theorems.
In Section~\ref{sec:center_point} we formalize the notion of depth, and 
review centerpoint theorems that guarantee the existence of points of large depth.
In Section~\ref{sec:projection} we present a new notion of depth together
with its centerpoint theorem (Theorem~\ref{thm:center_point_projection}). Using these results we prove the Depth
Carath\'eodory's theorem (Theorem~\ref{thm:car_depth}) in Section~\ref{sec:car_depth}.
In Section~\ref{sec:car_depth_plane} we prove a stronger planar version
of the Depth Carath\'eodory's theorem for points of small depth. 
Finally in Section~\ref{sec:helly_depth},
we prove depth versions of Helly's and Kirchberger's theorems.

\subsection{Helly's and Kirchberger's Theorems}\label{sec:helly}

\begin{theorem}[\textbf{Helly's theorem}]
 Let $\mathcal{F}=\{C_1,\dots,C_n\}$ be a family of $n \ge d+1$ convex sets in $R^{d}$.
 Suppose that for every choice of $(d+1)$ sets $D_i \in \mathcal{F}$ we have that $\bigcap_{i=1}^{d+1} D_i \neq \emptyset$.
 Then all of $\mathcal{F}$ intersects, that is $\bigcap \mathcal{F} \neq \emptyset$.
\end{theorem}
\begin{theorem}[\textbf{Kirchberger's theorem}~\cite{kirchberger}]
Let $R$ and $B$ be finite sets of points in $\mathbb{R}^d$. 
Suppose that for every subset $S \subset R \cup B$ of $d+2$ points, the set $S \cap R$
can be separated from $S \cap B$ by a hyperplane. Then $R$ can be separated
from $B$ by a hyperplane. 
\end{theorem}

Helly's theorem was discovered in 1913 by Helly, but he did not published it until 1923~\cite{helly}.
By then proofs by Radon~\cite{radon_helly} (1921) and by K\"onig~\cite{konig_helly} (1922)
had been published.
Carath\'eodory's and Helly's theorems are closely related in the sense that without
much effort each can be proven assuming the other (see for example Eggleston's book~\cite{convex_book});
and from either one of them one can prove Kirchberger's theorem. 
The link between Helly's and Carath\'eodory's theorems is given by the following lemma (see~\cite{convex_book}).

\begin{lemma}\label{lem:link}
 Let $x_1,\dots, x_n$ be points in $\mathbb{R}^d$;  let $H_q(x_i)$ be the halfspace that does not contain $q$, and that is
bounded by the following hyperplane: the hyperplane through $x_i$ and perpendicular to the line passing through $x_i$ and $q$.
 Then $\bigcap_{i=1}^n H_q(x_i)$ is empty if and only if $q \in \conv(\{x_1,\dots,x_n\})$.
\end{lemma}
\begin{proof}{}\

 $\Rightarrow )$ Suppose that $q \notin \conv(\{x_1,\dots,x_n\})$. Then there exist a hyperplane $\Pi$ 
 that separates $q$ from $\conv(\{x_1,\dots,x_n\})$. Let $\ell$ be the halfline with apex $q$, perpendicular
 to $\Pi$, and that intersects $\Pi$. Then for all $H_q(x_i)$, all but a finite segment
 of $\ell$ is contained in $H_q(x_i)$. This implies that $\bigcap_{i=1}^n H_q(x_i)$ is non-empty.
 
 $\Leftarrow )$ Suppose that $q \in \conv(\{x_1,\dots,x_n\})$. Then 
 there exist $\alpha_1,\dots, \alpha_{n}$, such that every $\alpha_i \ge 0$, 
\[\sum_{i=1}^{n} \alpha_i =1 \textrm{ and } \] 
\[q=\sum_{i=1}^{n} \alpha_i x_i.\] Note that for every $y \in H_q(x_i)$, we have that 
 \[(y-q)\cdot (x_i-q) \ge |q-x_i|^2>0. \] Suppose to the contrary that there exists
 a point $z \in \bigcap_{i=1}^n H_q(x_i)$. Since at least one of the $\alpha_i$ is greater
 than zero we have  
 \begin{eqnarray*}
  0 & = & (z-q)\cdot (q-q) \\
    & = & (z-q) \cdot \left ( \sum_{i=1}^{n} \alpha_i(x_i-q) \right ) \\
    & = &\sum_{i=1}^{n} \alpha_i (z-q) \cdot(x_i-q) >0,  
\end{eqnarray*}
a contradiction.
\end{proof}

Given the close relationship between Carath\'eodory's and Helly's theorems, whenever there is
a variant of one of them, it is natural to ask whether similar
variants exist of the other. We mention their colorful and fractional
variants.

In the classical versions of Helly's and Carath\'eodory's theorems, if every $(d+1)$-tuple of a family of objects
satisfies a certain property then all the family satisfies the property.
In the statements of the colorful versions of these theorems, the objects are  assigned one
of $d+1$ colors. A \emph{colorful tuple} is a tuple of objects containing one object of each available color.
The hypothesis is that the colorful tuples satisfy the same property as in the classical versions
and the conclusion is that one of the color classes satisfies the property. The colorful versions are the following.

\begin{theorem}[\textbf{Colorful Carath\'eodory's theorem~\cite{colorful_caratheodory}}]
 Let $X$ be a finite set of points in $\mathbb{R}^d$, such that each point
 is assigned one of $d+1$ colors. Let $q$ be a point in $\mathbb{R}^d$ such that
 every colorful tuple of points of $X$ does not contain $q$ in its convex hull.
 Then one of the color classes of $X$ does not contain $q$ in its convex hull.
\end{theorem}

\begin{theorem}[\textbf{Colorful Helly's theorem}]
Let $\mathcal{F}$ be a finite family of convex sets in $\mathbb{R}^d$, such that each set is assigned
one of $d+1$ colors. Suppose that every colorful tuple of convex sets in $\mathcal{F}$ has a non-empty intersection. Then
a color class of $\mathcal{F}$ has a non-empty intersection. 
\end{theorem}

The Colorful Helly's theorem was first proved by Lov\'asz(see~\cite{colorful_caratheodory});
the Colorful Carath\'edory's theorem was proved by B\'ar\'any.
A stronger version of the Colorful Carath\'edory theorem was proven independently by Holmsen, Pach, and Tverberg~\cite{points_surrounding},
and by Arocha, B\'ar\'any, Bracho, Fabila-Monroy and Montejano~\cite{very_colorful}.
In that version the hypothesis is the same and the conclusion  is that there exist
two color classes such that $q$ is not contained in the convex hull of their union. An even
stronger version was proved by Meunier and Deza~\cite{further_colorful}.

In the fractional versions, the hypothesis is that a large number of the $(d+1)$-tuples
satisfies the same property as the classical version and the conclusion is that a large
subfamily satisfies the property. The fractional version of Helly's theorem is as follows.

\begin{theorem}[\textbf{Fractional Helly's theorem~\cite{fractional_helly}}]
 For every $\alpha > 0$ and every dimension $d \ge 1$ there exists a constant $\beta=\beta(\alpha,d)>0$
  such that the following holds. Let $\mathcal{F}$ be a family
 of $n$ convex sets in $\mathbb{R}^d$. Suppose that at least $\alpha \binom{n}{d+1}$ of the 
 $(d+1)$-tuples of sets in $\mathcal{F}$ has a non-empty intersection. Then there exists a subfamily
 $\mathcal{F}' \subset \mathcal{F}$ of at least $\beta n$ sets with a non-empty intersection.
\end{theorem}
The Fractional Helly theorem was proved by Katchalski and Liu.

It is natural to ask whether a Fractional Carath\'eodory's theorem exists.
Interestingly enough, although the statement is true and well known, to the best of our knowledge there is no explicit 
mention of the Fractional Carath\'eodory's theorem in the literature.
The Fractional Carath\'eodory's theorem is just the relationship between the Tukey and the simplicial
 depth defined in Section~\ref{sec:center_point}. 
The statement is the following.

\begin{theorem}[\textbf{Fractional Carath\'eodory's theorem}]
For every $\alpha > 0$ and every dimension $d \ge 1$ there exists a constant $\beta=\beta(\alpha,d)>0$
  such that the following holds.  Let $X$ be a set of $n$ points in $\mathbb{R}^d$.
  Suppose that $q$ is a point such that at least $\alpha \binom{n}{d+1}$ of the $(d+1)$-tuples of points
  do not contain $q$ in their convex hull. Then there exists a subset $X'\subset X$ of at least
  $\beta n$ points such that its convex hull does not contain $q$.
\end{theorem}

 For recent variations and related results of Helly's theorem, see the survey of Amenta, De Loera, and Sober{\'o}n \cite{survey_helly_1}, and
 the survey by Holmsen and Wenger~\cite{survey_helly_2}.

\subsection{Centerpoint Theorems and Depth}\label{sec:center_point}

Let $X \subset \mathbb{R}^d$ be a set of $n$ points.
There are many formalizations of the notion 
of how deep a point $q \in \mathbb{R}^d$ is inside $\conv(X)$. We use two of them in this paper; 
they are defined as follows.

\begin{itemize}
 \item The \emph{Tukey depth} of $q$ with respect to $X$, is the minimum number of points of $X$
in every closed halfspace that contains $q$; we denote it by $\tilde{\tau}_X(q)$.

 \item  The \emph{simplicial depth} of $q$ with respect to $X$,  is the number
 of distinct subsets of $d+1$ points $x_1,\dots,x_{d+1}$ of $X$, such that $q$
 is contained in $\conv(\{x_1, \dots, x_{d+1}\})$; we denote it by $\tilde{\sigma}_X(q)$.
\end{itemize}
The word ``simplicial'' in the latter definition comes from the fact that $q$ 
is contained in the simplex with vertices $x_1,\dots,x_{d+1}$. In other words, $\tilde{\sigma}_X(q)$
is the number of simplices with vertices in $X$ that contain $q$ in their interior.
The two definitions are not equivalent (one does not determine the other). They are however related;
Afshani~\cite{afshani} has shown that $\Omega(n \tilde{\tau}_X(q)^d) \le \tilde{\sigma}_X(q) \le O(n^d \tilde{\tau}_X(q))$,
and that these bounds are attainable; Wagner~\cite{wagner} has shown a tighter lower bound of 
\[\tilde{\sigma}_X(q)\ge \frac{(d+1)\tilde{\tau}_X(q)^d n-2d\tilde{\tau}_X(q)^{d+1}}{(d+1)!}-O(n^d).\]

Both definitions have generalizations to Borel probability measures on $\mathbb{R}^d$;
let $\mu$ be such a measure.

\begin{itemize}
 \item The \emph{Tukey depth} of $q$ with respect to $\mu$, is the minimum of $\mu(H)$
 over all closed halfspaces $H$ that contain $q$; we denote it by $\tau_{\mu}(q)$.

 \item  The \emph{simplicial depth} of $q$ with respect to $\mu$,  is the probability that
 $q$ is in the convex hull of $d+1$ points chosen randomly and independently with distribution $\mu$;
 we denote it by $\sigma_\mu(q)$.
\end{itemize}

The Tukey depth was introduced by Tukey~\cite{tukey} and the simplicial depth by Liu~\cite{liu}.
Both definitions aim to capture how deep a point is inside a data 
set. As a result they have been widely used in statistical analysis.
For more information on these applications and other definitions of depth, see: The book
edited by Liu, Serfling and Souvaine~\cite{data_depth_book};  the survey by
Rafalin and Souvaine~\cite{survey_depth}; and the monograph by Mosler~\cite{mosler}.

For the purpose of this paper, we join the definitions for point sets and for probability measures by setting
$\tau_X(q):=\tilde{\tau}_X(q)/n$ and $\sigma_X(q):=\tilde{\sigma}_X(q)/\binom{n}{d+1}$.
We refer to them as the Tukey and simplicial depth of point $q$ with respect to $X$, respectively.
(Alternatively, note that a Borel probability measure is obtained from $X$ by defining
the measure of an open set $A \subset \mathbb{R}^d$ to be $|A\cap X|/n$.)
Throughout the paper, for exposition purposes, we present the proofs of our results for sets of points
rather than for Borel probability measures. However, we explicitly mention when such results
also hold for Borel probability measures.

Long before the concept of Tukey depth came about (1975), a point of large
Tukey depth was shown to always exists.  This was first proved
in the plane by Neumann~\cite{neumann} in 1945. Using Helly's theorem, Rado~\cite{rado} generalized this result to higher dimensions
in 1947. We rephrase this theorem in terms of the Tukey depth as follows. 

\begin{theorem}[\textbf{Centerpoint theorem for Tukey depth}] \label{thm:center_point}
 Let $X$ be a finite set of points in $\mathbb{R}^d$ (or a Borel probability measure
 on $\mathbb{R}^d$). Then there exists a point $q$ such that $\tau_X(q)\ge \frac{1}{d+1}$. 
\end{theorem}

A point satisfying Theorem~\ref{thm:center_point} is called
a \emph{centerpoint}. The bound of Theorem~\ref{thm:center_point} is tight;
there exist sets of $n$ points in $\mathbb{R}^d$ such that no point of $\mathbb{R}^d$
has Tukey depth larger than $\frac{1}{d+1}$ with respect to these point sets.

The preceding of a centerpoint theorem before the definition of Tukey
depth also occurred with the definition of simplicial depth (1990). Boros and F\"uredi~\cite{boros} proved
in $1984$ that if $X$ is a set of $n$ points in the plane, then
there exists a point $q$ contained in at least $\frac{2}{9}\binom{n}{3}$ of the 
triangles with vertices on $X$. This was generalized to higher dimensions
by B\'ar\'any~\cite{colorful_caratheodory} in 1982. This result also holds for Borel
probability measures (see Wagner's PhD thesis~\cite{wagner}); we rephrase it
in terms of the simplicial depth as follows.

\begin{theorem}[\textbf{Centerpoint theorem for simplicial depth}] \label{thm:center_point_simplicial}
 Let $X$ be a finite set of points in $\mathbb{R}^d$(or a Borel probability measure
 on $\mathbb{R}^d$). Then there exists a point $q$  and a constant $c_d >0$ (depending
 only on $d$) such that $\sigma_X(q)\ge c_d$. 
\end{theorem}

Theorem~\ref{thm:center_point_simplicial} has been named the ``First Selection Lemma'' by
Matou\v{s}ek~\cite{lectures}.
In contrast with Theorem~\ref{thm:center_point} the exact value of $c_d$ (for values of $d$ greater
than $2$) is far from known. In the case of when $X$ is a point set, the search for better  bounds for $c_d$ 
has been an active area of research. The current best upper bound (for point sets) is $c_d \le \frac{(d+1)!}{(d+1)^{d+1}}$.
This was proven by Bukh, Matou\v{s}ek and Nivasch~\cite{stabbing}; they conjecture
that this is the exact value of $c_d$. As for the lower bound B\'ar\'any's original proof yields
$c_d\ge \frac{1}{(d+1)^{d+1}}$. Using topological methods Gromov~\cite{gromov} has significantly improved
this bound to $c_d \ge \frac{2d}{(d+1)!(d+1)}$. Shortly after, Karasev~\cite{karasev} provided
a simpler version of Gromov's proof, but still using topological methods. Matou\v{s}ek and Wagner~\cite{on_gromov}
gave an expository account of the combinatorial part of Gromov's proof; they also
slightly improved Gromov's bound and showed limitations on his method.

\section{Projection Tukey Depth}\label{sec:projection}

Along the way of proving the Depth Carath\'eodory's theorem we prove a result (Theorem~\ref{thm:center_point_projection}) similar in spirit
to the centerpoint theorems above. Taking a lesson from history we first define a notion
of depth and then phrase our result accordingly. 

Assume that $d \ge 2$ and let $X$ be a set of $n$ points in $\mathbb{R}^d$. Let $q$ be a point of $\mathbb{R}^d\setminus X$. Given a  point $p \in \mathbb{R}^d$ distinct from $q$, 
let $r(p,q)$ be the infinite ray passing through $p$ and with apex $q$. 
For a set $A\subset \mathbb{R}^d$, not containing $q$, let $R(A,q):=\{r(p,q): p \in A\}$ be the set of rays with apex $q$ and passing through
a point of $A$. Let $\Pi$ be an oriented hyperplane containing $q$, and 
let $\Pi^+$ and $\Pi^-$ be  two hyperplanes, parallel to $\Pi$, strictly above and below $\Pi$, respectively.
Let $X^+:=\Pi^+\cap R(X,q)$ and $X^-:=\Pi^-\cap R(X,q)$. Let $L(q)$ be the set of straight lines
that contain $q$. See Figure~\ref{fig:proy}.

We define the \emph{projection Tukey depth} of $q$ with respect to $\Pi$ and $X$, to be
\[\pi_{X,\Pi}(q):=\max\{\min\{\tau_{X^+}(\ell \cap \Pi^+),\tau_{X^-}(\ell \cap \Pi^-)\}: \ell \in L(q)\}.\] 
Intuitively if $q$ has large projection Tukey depth with respect to $\Pi$ and $X$, 
then there exists a direction in which $q$ can be projected
to $\Pi^+$ and $\Pi^-$, such that the images of $q$ have large Tukey depth with respect to $X^+$ and $X^-$.  See Figure~\ref{fig:proy}.
Finally we define the \emph{projection Tukey depth} of $q$ with respect to $X$, as the minimum of this value over all $\Pi$. That
is, \[\pi_X(q):=\min\{\pi_{X,\Pi}(q): \Pi \textrm{ is a hyperplane containing } q\}.\]

\begin{figure}
  \begin{center}
   \includegraphics[width=0.8\textwidth]{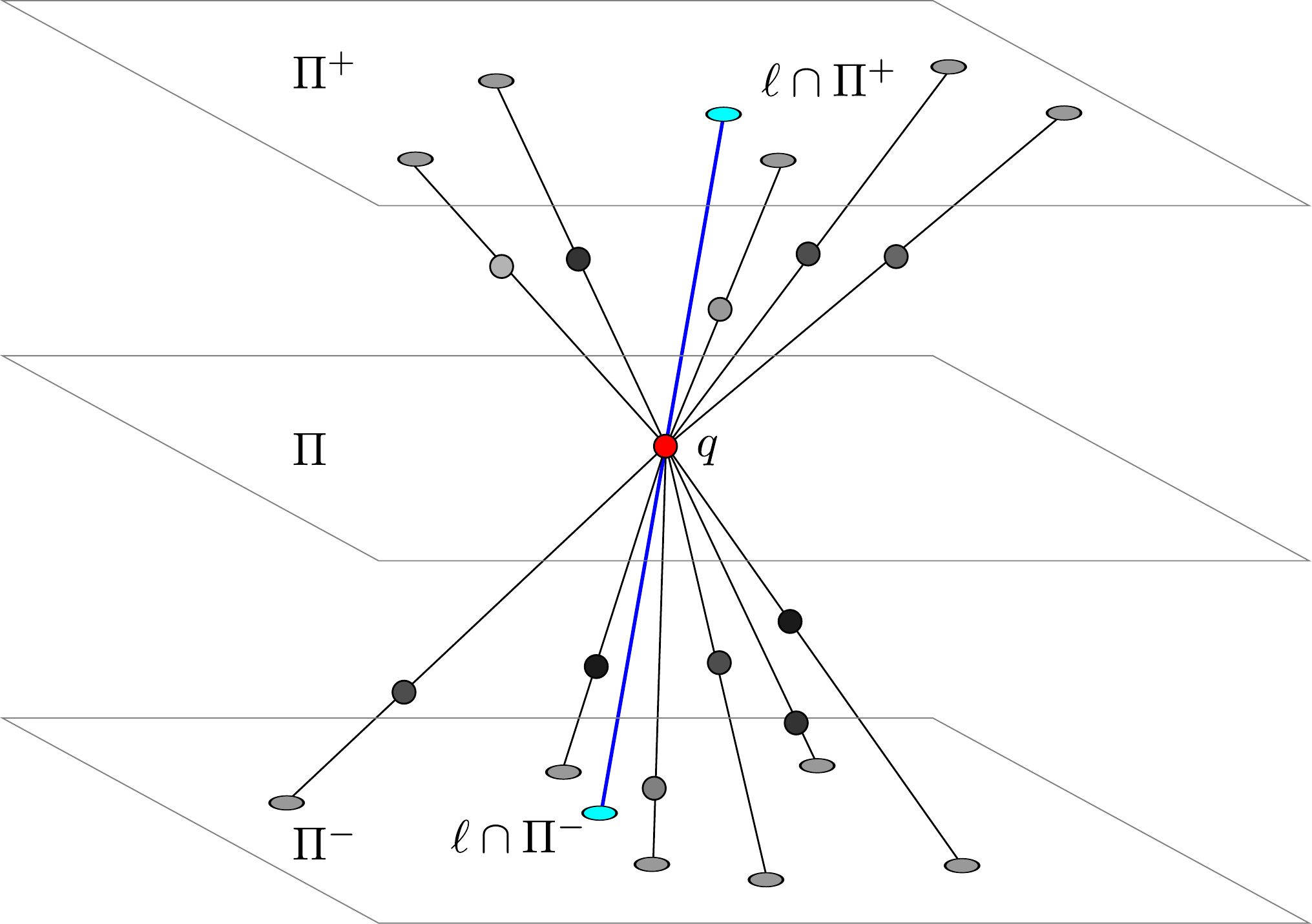}
\end{center}
\caption{Depiction of $q$, $X$, $X^+$, $X^-$, $\Pi$, $\Pi^+$, and  $\Pi^-$. }
\label{fig:proy}
\end{figure}

To provide the definition for Borel probability measures, we use $\mu$ to define
Borel measures $\mu^+$ and $\mu^-$ on $\Pi^+$ and $\Pi^-$, respectively. Let 
$U$ and $D$ be the set of points of $\mathbb{R}^d$ that lie above and below
$\Pi$, respectively. Let $\mu^+(A):=\mu(R(A,q))/\mu(U)$ for sets $A\subset \Pi^+$,
and $\mu^-(A):=\mu(R(A,q))/\mu(D)$ for sets $A\subset \Pi^-$. The \emph{projection Tukey depth} of $q$ with respect to
$\Pi$ and $\mu$, and the \emph{projection Tukey depth} of $q$ with respect to $\mu$ are defined respectively as
\[\pi_{\mu,\Pi}(q):=\max\{\min\{\tau_{\mu^+}(\ell \cap \Pi^+),\tau_{\mu^-}(\ell \cap \Pi^-)\}: \ell \in L(q)\}\textrm{ and}\]
\[\pi_\mu(q):=\min\{\pi_{\mu,\Pi}(q): \Pi \textrm{ is a hyperplane containing } q\}.\]

The following lemma lower bounds the projection Tukey depth of $q$  in terms of its Tukey depth.

\begin{lemma}\label{lem:projection}
 Let $X$ be a set of $n$ points in $\mathbb{R}^d$ (or a Borel probability measure
 on $\mathbb{R}^d$) and let $q$ be a point in $\mathbb{R}^d$. Then $\pi_X(q) \ge \min\{\tau_X(q),\frac{1}{d}\}$. 
\end{lemma}
\begin{proof}
Let $\delta:= \min \left \{\tau_X(q),\frac{1}{d} \right \}$ and $\varepsilon > 0$.
We prove the result by showing that $\pi_X(q)\ge \delta-\varepsilon$.
Let $\Pi$ be a hyperplane containing $q$; define $\Pi^+$, $\Pi^-$, $X^+$ and $X^-$ with respect to $\Pi$, as above.
We look for a straight line $\ell$ passing through $q$ such that 
\[\tau_{X^+}(\ell \cap \Pi^+)\ge \delta-\varepsilon \textrm{ and }  \tau_{X^-}(\ell \cap \Pi^-) \ge \delta-\varepsilon.\]
For this we find a point $q^+$ in $\Pi^+$, with certain properties,  such that the straight line passing through $q$ and $q^+$ is 
the desired $\ell$. 

For each point $p\in X^{-}$, let $p'$ be intersection point of the line passing
through $q$ and $p$ with $\Pi^+$; let $X'$ be the set of all such
$p'$. Note that $X'$ is just the reflection of $X^{-}$ through $q$ into $\Pi^{+}$. 
Let $\mathcal{P}_q$ be the set of hyperplanes passing through $q$ and not parallel to $\Pi$,
and let $\mathcal{P}_+$ be the set of $(d-2)$-dimensional flats contained in $\Pi^{+}$.
There is a natural  one-to-one correspondence between $\mathcal{P}_+$ and $\mathcal{P}_q$.
For each $l \in \mathcal{P}_+$,  let $\pi_l$ be the hyperplane in $\mathcal{P}_q$ containing both $q$ and $l$;
conversely for each $\pi \in \mathcal{P}_q$, let $l_\pi$ 
be the $(d-2)$-dimensional flat in $\mathcal{P}_+$ defined by the intersection of $\pi$ and $\Pi^{+}$. 
 
Note that the following relationship holds for every
pair of points $p_1 \in X^{+}$ and $p_2 \in X^{-}$;

\MyQuote{$p_1$ and $p_2$ are on the same
half-space defined by  a $\pi \in P_q$ if and only if $p_1$ and $p_2'$ are in the opposite half-spaces
of $\Pi^{+}$ defined by $l_\pi$.} We use this observation
to find $q^+$.
 
Let $\mathcal{H}$ be the set of all $(d-1)$-dimensional half-spaces of $\Pi^{+}$ that contain more than 
$|X^{+}|-(\delta-\varepsilon)|X^{+}|$ points of $X^{+}$; 
let $\mathcal{H}'$ be the set of  all $(d-1)$-dimensional half-spaces of $\Pi^{+}$ that contain more than 
$|X'|-(\delta-\varepsilon)|X'|$ points of $X'$. Therefore, since $\delta-\varepsilon < \frac{1}{d}$, a centerpoint of $X^+$, given
by Theorem~\ref{thm:center_point}, is contained
in every halfspace in $\mathcal{H}$; otherwise, we obtain a contradiction since the opposite half-space
would contain the centerpoint and less than $\frac{1}{d}|X^{+}|$ points of $X^{+}$. Likewise a centerpoint of $X'$, given
by Theorem~\ref{thm:center_point}, is contained
in every halfspace in $\mathcal{H}'$. Therefore,  $Q:=\bigcap \mathcal{H}$ and
$Q':=\bigcap \mathcal{H}'$ are non-empty. A point in the intersection
of $Q$ and $Q'$ is our desired $q^+$.

For the sake of a contradiction, suppose that $Q$ and $Q'$ are disjoint. Let $l \in \mathcal{P}_+$  be a $(d-2)$-dimensional flat 
that separates them in $\Pi^{+}$. Let $h$ be the halfspace (in $\mathbb{R}^d$) defined by $\pi_l$ that 
contains $Q'$ and does not contain $Q$. Note that $h$ contains at least $|X'|-(\delta-\varepsilon)|X'|$ points
of $X'$ and at most $(\delta-\varepsilon)|X^{+}|$ points of $X^{+}$. By (2), $h$ contains at most $(\delta-\varepsilon)|X^{-}|$ points of $X^{-}$. 
Therefore, $h$ contains at most $(\delta-\varepsilon)|X^{+}|+(\delta-\varepsilon)|X^{-}|=(\delta-\varepsilon)|X|$ points
of $X$---a contradiction. Therefore, $Q$ and $Q'$ intersect.

Let $q^{+}$ be a point in  $Q \cap Q'$ and let $\ell$ be the straight line passing through
$q$ and $q^{+}$; let $q^-:=\ell \cap \Pi^-$. Note that $q^{-}$ is 
in the intersection of all $(d-1)$-dimensional halfspaces that contain more than
$|X^{-}|-(\delta-\varepsilon)|X^{-}|$ points of $X^{-}$.  We have that every halfspace in $\Pi^{+}$
that contains $q^{+}$, contains at least $(\delta-\varepsilon)|X^{+}|$ points of $X^{+}$, and
every half space in $\Pi^{-}$ that contains $q^{-}$, contains at least $(\delta-\varepsilon)|X^{-}|$ points of $X^{-}$;
the result follows.
\end{proof}

Although Lemma~\ref{lem:projection} bounds the projection Tukey depth of $q$ with respect to $X$ in terms
of its Tukey depth, it does so up to 
a point; when the Tukey depth is larger than $\frac{1}{d}$, the best lower bound on
the projection Tukey depth given by Lemma~\ref{lem:projection} is of $\frac{1}{d}$;
this bound can be tight. Suppose that $X$ is such that $X^-$ is
equal to the reflection of $X^+$ through $q$ into $\Pi^-$.
Moreover, assume that $X^+$ is such that every point in $\Pi^+$ has Tukey depth of at most $\frac{1}{d}$
with respect to $X^+$.
Note that in this case the Tukey depth of $q$ with respect to $X$ is $\frac{1}{2}$ and the projection Tukey depth of $q$
with respect to $\Pi$ and $X$ is at most $\frac{1}{d}$. The latter implies that the projection
Tukey depth of $q$ with respect to $X$ is at most $\frac{1}{d}$

Lemma~\ref{lem:projection} and Theorem~\ref{thm:center_point} yield at once
a centerpoint theorem for the projection Tukey depth. 

\begin{theorem}[\textbf{Centerpoint theorem for projection Tukey depth}] \label{thm:center_point_projection}
 Let $X$ be a finite set of points in $\mathbb{R}^d$ (or a Borel probability measure
 on $\mathbb{R}^d$). Then there exists a point $q$ such that $\pi_X(q)\ge \frac{1}{d+1}$. 
\end{theorem}
\begin{proof}
 By Theorem~\ref{thm:center_point} there exists a point $q$ of Tukey depth with respect
 to $X$ of at least $\frac{1}{d+1}$. By Lemma~\ref{lem:projection}, $\pi_X(q) \ge \min\{\tau_X(q),\frac{1}{d}\} \ge \frac{1}{d+1}$.
\end{proof}

\section{Depth Carath\'eodory's Theorem}\label{sec:car_depth}

Using Lemma~\ref{lem:projection} it can be shown by induction on $d$ that if $q$ has 
a large Tukey depth with respect to $X$ and $|X|$ is sufficiently large with respect
to $d$, then there exist large subsets $X_{1},\dots,X_{2^d}$ of $X$, such that 
for every choice of points
$x_i \in X_i$, $q$ is contained in
$\conv(\{x_1,\dots, x_{2^d}\})$. To reduce this number of subsets to $d+1$ we need a result from~\cite{same_type}.

The order type is a combinatorial abstraction
of the geometric properties of point sets. They were introduced by Goodman and Pollack in~\cite{order_type}.
Two sets of points $X$ and $X'$ in $\mathbb{R}^d$ are said to have the same 
\emph{order type} if there is a bijection, $\varphi$, between them that satisfies the following.
The orientation of every $(d+1)$-tuple $(x_1,\dots,x_{d+1})$ of points of $X$ is
equal to the orientation of the corresponding $(d+1)$-tuple  $(\varphi(x_1),\dots, \varphi(x_{d+1}))$ of $X'$.
This means that the signs of the determinants 
$\operatorname{det}\left [\binom{x_1}{1},\dots,\binom{x_{d+1}}{1} \right ]$ and
$\operatorname{det}\left [\binom{\varphi(x_1)}{1},\dots,\binom{\varphi(x_{d+1})}{1} \right ]$ are equal.
Let  $\boldsymbol{x}:=(x_1,\dots,x_{m})$ and $\boldsymbol{y}:=(y_1,\dots,y_{m})$ be
two $m$-tuples of $\mathbb{R}^d$ (for $m\ge d+1$). We say that
$\boldsymbol{x}$ and $\boldsymbol{y}$ have the same order type
if for every subsequence $(i_1,\dots ,i_{d+1})$ of $(1,\dots, m)$, the orientation
of the $(d+1)$-tuple $(x_{i_1},\dots,x_{i_{d+1}})$ of $\boldsymbol{x}$ is the same
as the orientation of the $(d+1)$-tuple $(y_{i_1},\dots,y_{i_{d+1}})$ of $\boldsymbol{y}$.

In particular this implies that if a
point $q\in \mathbb{R}^d$ is such that $\boldsymbol{x}:=(x_1,\dots, x_m,q)$ and $\boldsymbol{x}':=(x_1',\dots, x_m',q)$ have the same
order type, then $q \in \conv(\{x_1,\dots, x_m\})$ if and only if it $q \in \conv(\{x_1',\dots, x_m'\})$.
B\'ar\'any and Valtr~\cite{same_type} proved the following theorem on order types of tuples
of point sets.

\begin{theorem}[\textbf{Same-type lemma}]\label{thm:same_type}
 For any integers $d,m \ge 1$, there exists a constant $c'=c'(d,m)>0$ (that
 depends only on $d$ and $m$) such that the following holds. 
 Let $X_1, X_2,\dots, X_m$ be finite sets of points in $\mathbb{R}^d$ 
 (or Borel probability measures on $\mathbb{R}^d$). Then there exist
 $Y_1 \subseteq X_1,\dots, Y_m \subseteq X_m$, such that every pair
 of $m$-tuples $(z_1,\dots, z_m)$ and $(z_1',\dots, z_m')$ with $z_i, z_i' \in Y_i$
 have the same order type. Moreover for all $i=1,2,\dots m$, $|Y_i|\ge c'|X_i|$.
 \end{theorem}
 
We note that the Same-type lemma is phrased only for points in general position in
both \cite{same_type} and in Matou\v{s}ek's book~\cite{lectures}. However, in Remark 5 of~\cite{same_type}
it is noted that the result holds for Borel probability measures and for points not
in general position. 

We are ready to prove the Depth Carath\'eodory's theorem.
\begin{proof}[Proof of Theorem~\ref{thm:car_depth}]
 The result holds for $d=1$ with $c(1,\tau_X(q))=\tau_X(q)$. Assume that $d>1$ and proceed by 
 induction on $d$. 
 Let $n:=|X|$ and let $\Pi$ be a hyperplane containing $q$ that bisects $X$.
 This is, on both of the open halfspaces defined by $\Pi$ there are at least
 $\lfloor n/2 \rfloor$ points of $X$. Define $X^+, X^-, \Pi^+$ and $\Pi^-$ 
 as in Section~\ref{sec:projection} with respect to $X$ and $\Pi$.
 
 Let $\delta:=\min\{\tau_X(q), \frac{1}{d}\}$. By Lemma~\ref{lem:projection} the projection
 Tukey depth of $q$ with respect to $X$ is at least $\delta$. Therefore, there exist
 a line $\ell$ such that $q^+:=\ell \cap \Pi^+$ and $q^-:=\ell \cap \Pi^-$ have Tukey depth 
 at least $\delta$ with respect to $X^+$ and $X^-$ respectively. By induction there exists 
 a constant $c(d-1,\delta)$ and sets $Y_1^+,\dots Y_d^+ \subset X^+$ and $Y_1^-,\dots Y_d^- \subset X^-$
 such that the following holds. Every $Y_i^+$ has cardinality at least
 $c(d-1,\delta)|X^+| \ge c(d-1,\delta) \lfloor n/2 \rfloor$, and every $Y_i^-$ has cardinality at least
 $c(d-1,\delta)|X^-| \ge c(d-1,\delta) \lfloor n/2 \rfloor$; moreover, $q^+\in \conv(\{x_1,\dots, x_d\})$ for every choice of $x_i \in Y_i^+$,
 and $q^- \in \conv(\{x_1',\dots, x_d'\})$ for every choice of $x_i' \in Y_i^-$. 
 
 \MyQuote{Therefore, $q$ is in the convex hull of 
 $\{x_1,\dots, x_d\} \cup \{x_1',\dots, x_d'\}$ for every choice of $x_i \in Y_i^+$ and $x_i' \in Y_i^-$.}
 
 Apply the Same-type lemma to $Y_1^+,\dots, Y_d^+, Y_1^-,\dots, Y_d^-,\{q\}$, and obtain sets
 $Z_1\subset Y_1^+,\dots, Z_d\subset Y_d^+,$ and $Z_{d+1}\subset Y_1^-,\dots, Z_{2d}\subset Y_d^-$
 each of at least $c'(d,2d)c(d-1,\delta) \lfloor n/2 \rfloor$ points, such that the following holds.
 Every pair of $(2d+1)$-tuples $(z_1,\dots, z_{2d},q)$ and $(z_1',\dots, z_{2d}',q)$ with $z_i \in Z_i$ and $z_i' \in Z_i'$
 have the same order type. Let $(z_1,\dots, z_{2d},q)$ be one such $(2d+1)$-tuple. By (2), $q$ is in the convex hull of $\{z_1, \dots, z_{2d}\}$.
 Therefore, by Carath\'eodory's theorem
 there exists a $(d+1)$-tuple $(i_1,\dots,i_{d+1})$ such that $q$ is a convex combination
 of $z_{i_1},\dots, z_{i_{d+1}}$. The result follows by setting $X_j:=Z_{i_j}$.
\end{proof}

Note that Theorem~\ref{thm:car_depth} also applies for the simplicial depth. That is, suppose that $q$ is in a constant
proportion of the simplices spanned by points of $X$. Then, there exist subsets $X_1,\dots,X_{d+1}$ of $X$,
of linear size, such that $q$ is in every simplex that has exactly one vertex in each $X_i$.

It is noted in~\cite{same_type} that the constant in Theorem~\ref{thm:same_type} is bounded from below by 
\[c'(d,m) \ge (d+1)^{-(2^d-1)\binom{m-1}{d}}.\]
Therefore, the proof of Theorem~\ref{thm:car_depth} implies that $c(d,\tau_X(q))$ is an increasing function on $\tau_X(q)$,
when $0 < \tau_X(q)\le \frac{1}{d}$ and $d$ fixed.

\subsection{Depth Carath\'eodory's Theorem in the Plane}\label{sec:car_depth_plane}

The Depth Carath\'eodory's theorem~(Theorem~\ref{thm:car_depth}) can be applied when $q$
has constant Tukey depth with respect to $X$. That is when $\tau_X(q)=c$ for some 
positive constant $c$. In  this section we prove the Depth Carath\'eodory's
theorem in the plane for points of subconstant depth with respect to $X$, for example
when $\tau_X(q)=\frac{1}{n}$. We use the simplicial depth as it is easier
to quantify the depth of a point in this case. Also, we revert to using the simplicial depth of
$q$ with respect to $X$  as the number of triangles $\tilde{\sigma}_X(q)$ with vertices on $X$
that contain $q$ (rather than this number divided by $\binom{n}{2}$). 
We consider only the case when $X$ is a set of $n$ points in general position
in the plane. 

\begin{figure}
  \begin{center}
   \includegraphics[width=0.6\textwidth]{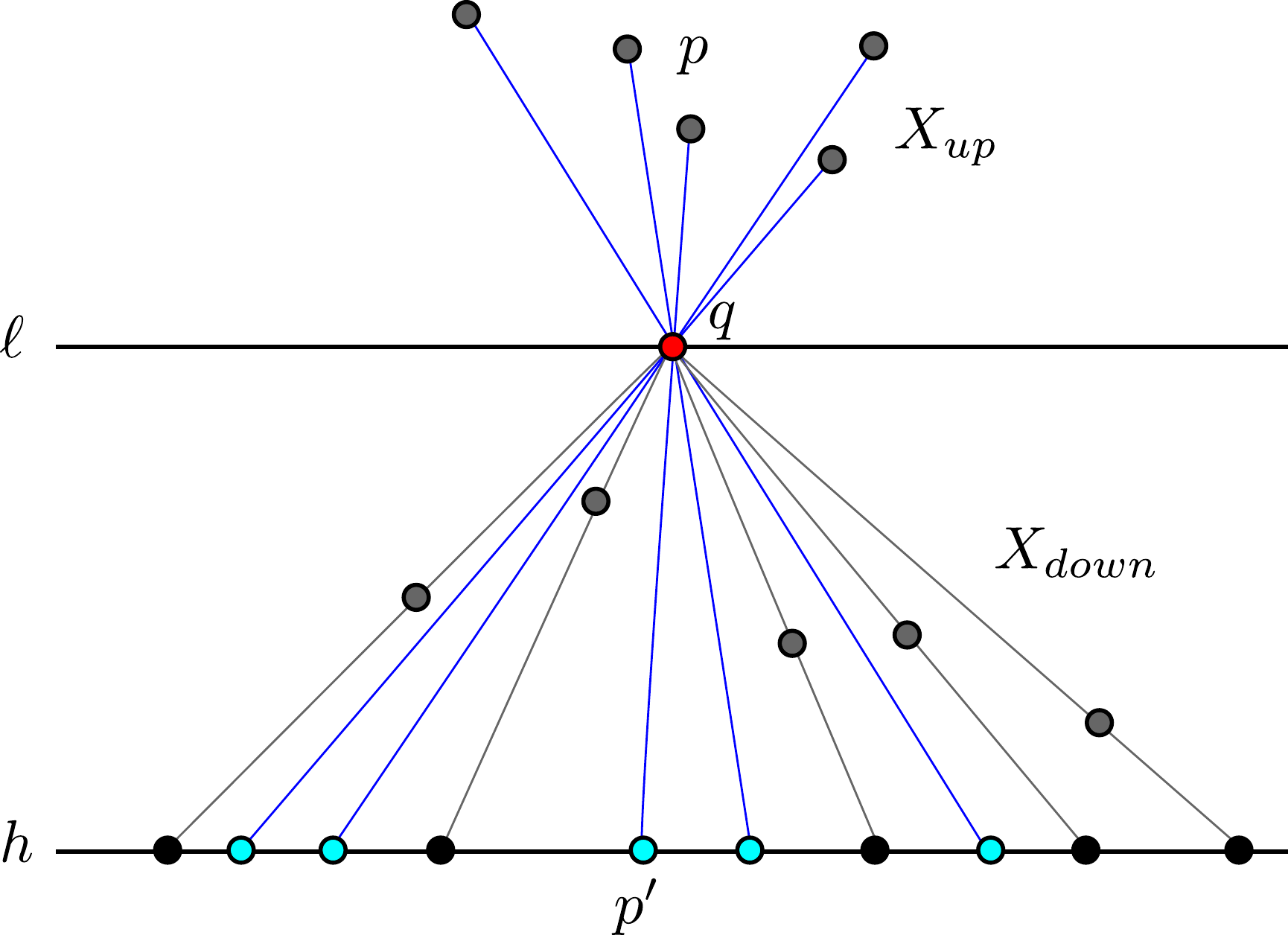}
\end{center}
\caption{Illustration of the projection in the proof of Theorem~\ref{thm:car_depth_plane}. 
$X_{up}'$ is represented as blue points on the line $h$, and $X_{down'}$ as black points.}
\label{fig:proy2}
\end{figure}

\begin{theorem}\label{thm:car_depth_plane}
Let $X$ be a set of $n$ points in general position in the plane. Let $q\in \mathbb{R}^2 \setminus X$ be a point such that
$X\cup \{q\}$ is in general position.
Then $X$ contains three disjoint subsets $X_1$, $X_2$, $X_3$ such that $|X_1||X_2||X_3| \ge \frac{1}{16 \ln n}\tilde{\sigma}_X(q)$,
and $q$ is contained in every triangle with vertices $x \in X_1$, $y \in X_2$, $z \in X_3.$
\end{theorem}
\begin{proof}
Let $\ell$ be a line passing through $q$ that bisects $X$. Without loss of generality, assume that at least half of the triangles which contain $q$ and
which have their vertices in $X$, have two of their vertices below $\ell$;  denote this set of triangles by $T$. 
Further assume that $\ell$ is horizontal.
 Let $X_{down}$ be the points of $X$ below $\ell$ and let $X_{up}$ be the points of $X$ above
$\ell$. Project $X$ on a horizontal line $h$ far below $X_{down}$ as follows. The image $p'$ of a point $p \in X$ is the intersection
point of the line through $q$ and $p$ with $h$. Let $X_{up}'$ and $X_{down}'$ be the images of 
$X_{up}$ and $X_{down}$, respectively. See Figure~\ref{fig:proy2}.

Note that a triangle with vertices $x \in X_{down}$, $y\in X_{up}$ and $z \in X_{down}$ contains $q$ 
if and only if $x'$, $y'$ and $z'$ appear in the order $(x',y',z')$ with respect to $h$. For a point $p' \in X_{up}'$
let $l(p')$ be the number of points in $X_{down}'$ to its left, and let 
$r(p')$ be the number of points in $X_{down}'$ to its right. By the previous observations we have that
\[|T|=\sum_{p' \in X_{up}'} l(p') r(p').\]

Let $\mathcal{L}:=\{p' \in X_{up}': r(p') \ge l(p')\}$ and $\mathcal{R}:=\{p' \in X_{up}': r(p')< l(p')\}$.
Assume without loss of generality that at least half of the triangles in $T$ are such that one of its vertices lies
in $\mathcal{L}$. That is
\begin{equation}\label{eq:lr}
\sum_{p' \in \mathcal{L}} r(p') l(p') \ge \frac{|T|}{2}.
\end{equation}
Also note that
\begin{equation}\label{eq:sum_l}
\sum_{p' \in \mathcal{L}} r(p') l(p') \le n\sum_{p' \in \mathcal{L}}l(p').
\end{equation}

Let $p_1',\dots, p_m'$ be the points in $\mathcal{L}$ in their left-to-right order in $h$. Consider
the sum \[\sum_{i=1}^{m} l(p_i')(m-i+1).\] Let $M$ be the maximum value attained by a term in this sum.
Then, for all $i=1,\dots, m$, we have that $l(p_i')\le M \frac{1}{(m-i+1)}$. Combining this observation
with (\ref{eq:lr}) and (\ref{eq:sum_l}), we have
\[ M \ln n \ge M \sum_{i=1}^m \frac{1}{(m-i+1)} \ge \sum_{i=1}^{m} l(p_i') \ge \frac{|T|}{2n}.\]
Therefore, 
\begin{equation}
 M \ge \frac{|T|}{2n\ln n}.
\end{equation}

Let $i^*$ be such that $l(p_{i^*}')(m-i^*+1)=M$. Set: $X_1$ to be the set of points of $X_{down}$ such their images
are to the left of $p_{i^*}$, $X_2$ to be set of points $\{p_{i^*},\dots,p_{m}\}$ of $X_{up}$
whose images lie between $p_{i^*}'$ and $p_{m}'$, 
and $X_3$ the set of set of points of $X_{down}$ such their images
are to the right of $p_m'$. Note that every triangle with vertices $x \in X_1$, $y \in X_2$, $z \in X_3$, contains
$q$. Also note that $|X_{down}| \ge n/2$ and $r(p') \ge  l(p')$ for all $p' \in \mathcal{L}$ imply that $|X_3| \ge n/4$. 
Moreover, \[|X_1||X_2||X_3| \ge l(p_{i^*}')(m-i^*+1)\frac{n}{4} = M\frac{n}{4}\ge \frac{1}{8 \ln n}|T| \ge \frac{1}{16 \ln n}\tilde{\sigma}_X(q);\]
the result follows.
\end{proof}

\section{Helly's and Kirchberger's Depth Theorems}~\label{sec:helly_depth}

In this section we use the Depth Carath\'eodory's theorem to prove ``depth'' versions of Helly's and  Kirchberger's theorems. 
Afterwards, for the sake of completeness we show that the Depth Helly's theorem
implies the Depth Carath\'eodory theorem for point sets. 

\begin{theorem}[\textbf{Depth Helly's theorem}]\label{thm:helly_depth}
 Let $ 0 \le \beta \le 1$ and $\alpha=c(d,1-\beta)$ (where $c(d,\beta)$ is as in Theorem~\ref{thm:car_depth}). 
 Let $\mathcal{F}=\{C_1,\dots, C_n\}$ be a family of $n \ge d+1$ convex sets in $\mathbb{R}^d$. Suppose that for every
 choice of subfamilies $F_1,\dots,F_{d+1}$ of $\mathcal{F}$, each with at least $\alpha n$
 sets, there exists a choice of sets $D_i \in F_i$ such that $\bigcap_{i=1}^{d+1} D_i \neq \emptyset$.
 Then there exists a subfamily $\mathcal{F}'\subset \mathcal{F}$ of at least $\beta n$ sets such that
 $\bigcap \mathcal{F}'\neq \emptyset$.
\end{theorem}
\begin{proof}
We assume that the $n$ convex sets are compact. The result for non-compact convex sets follows easily.
 Suppose for a contradiction that every subfamily of $\mathcal{F}$ of at least $\beta n$ sets
 is not intersecting. For every  $\mathcal{F}'$ subfamily of 
$\mathcal{F}$ of at least $\beta n$ sets and every $x \in \mathbb{R}^d$ define 
\[f(x,\mathcal{F}')=\max\{d(x,C):C \in \mathcal{F}'\},\]
where $d(x,C)$ denotes the distance from point $x$ to the set $C$. 
By hypothesis $f(x, \mathcal{F}') > 0$.
Since the elements of $\mathcal{F}$ are closed and bounded, $f(x, \mathcal{F}')$ 
attains a minimum value at some point of $\mathbb{R}^d$. 
Let 
\[f(x)=\min\{f(x,\mathcal{F}'):\mathcal{F}'\subset \mathcal{F} \textrm{ and } |\mathcal{F}'| \ge \beta n\}\]
Since there are a finite number of subfamilies of $\mathcal{F}$ and the previous remark, $f(x)$ attains its
minimum value at some point of $q \in \mathbb{R}^d$.  

For every $C_i \in \mathcal{F}$,
let $x_i$ be a point in $C_i$ such that $d(q,C_i)=d(q,x_i)$. 
Let $X$ be the multiset with elements $x_1,\dots,x_n$. Let $C_1',\dots,C_k'$ be the sets in $\mathcal{F}$
such that $f(q)=d(q,C_i')$ and let $x_1',\dots, x_k'$ be their corresponding points
in $X$.  We first show that 

\MyQuote{$q$ is contained in $\conv (\{x_1',\dots,x_k'\})$.}

Suppose for a contradiction  that $q$ is not contained in $\conv (\{x_1',\dots,x_k'\})$. 
Let $q'$ be a point closer to $\conv (\{x_1',\dots,x_k'\})$ than $q$. 
Let $\mathcal{F}'$ be a subfamily of $\mathcal{F}$ of at least $\beta n$ sets 
such that $f(q,\mathcal{F}')=f(q)$. Suppose that $C$ is a set of $\mathcal{F}'$ such
that $d(q,C)<f(q,\mathcal{F}')=f(q)$. Note that we can choose $q'$ sufficiently close
to $q$ so that $d(q',C)<f(q)$. We choose $q'$ sufficiently close to $q$ 
so that $d(q',C)<f(q)$ for all such $C \in \mathcal{F}'$. Let now $C$ 
be a set of $\mathcal{F}$ such that $d(q,C)=f(q,\mathcal{F}')=f(q)$.
Since the closest point of $C$ to $q$ is one of the $x_i'$, then
$d(q',C)<d(q,C)$. Therefore, $f(q',\mathcal{F}')<f(q,\mathcal{F})$ and $f(q') < f(q)$---a contradiction to our choice of $q$.

To every point $x_i \in X$ assign the weight $\frac{m(x_i)}{n}$, where
$m(x_i)$ is the multiplicity of $x_i$ in $X$. 
Thus, we may regard $X$ as a Borel probability measure.
We now show that 

\MyQuote{$q$ has Tukey depth greater than $1-\beta$ with respect to $X$.}

Suppose to the contrary that there exists a hyperplane $\Pi$ through $q$, such that the measure of one of the two halfspaces bounded by
$\Pi$ is at most $(1-\beta)$. Then on the opposite side there is a subset $X'$ of $X$ of measure 
at least $\beta$. By (6) at least one of these points must be in $\{x_1',\dots,x_k'\}$. If we move
$q$ slightly closer to $\conv(X')$ we obtain a $q'$ with $f(q')< f(q)$; this is a contradiction to the choice of $q$.

Apply the Depth Carath\'eodory theorem to $X$ and $q$, and obtain $X_1, \dots, X_{d+1}$ subsets
of $X$ each of at least $\alpha n$ points, such that the following holds. For every choice of $y_i$ in $X_i$,
$q$ is contained in $\conv(\{y_1,\dots, y_{d+1}\})$.  Let $C_i''$ be the set in $\mathcal{F}$ that defines
$y_i$.  Then by Lemma~\ref{lem:link}, \[\bigcap_{i=1}^{d+1} C_i ''=\emptyset.\]
Let \[F_i=\{C_j \in \mathcal{F}: \textrm{ there is a point } y \textrm{ in } X_i \textrm{ defined by } C_j\}\]
Since each $X_i$ is a multiset and every $x_i'$ is associated to a different $C_i$, then $F_1,\dots,F_{d+1}$ are subsets of $\mathcal{F}$, each with at least $\alpha n$
 sets, such that for every choice  $D_i \in F_i$, $\bigcap_{i=1}^{d+1} D_i$ is empty---a contradiction.
\end{proof}

\begin{proof}[Proof of Theorem~\ref{thm:car_depth} (for point sets) using Theorem~\ref{thm:helly_depth}]\

Let $X:=\{x_1,\dots,x_n\}$ be a set of $n$ points in $\mathbb{R}^d$ and let $q$ be a point in $\mathbb{R}^d$. 
Let $\beta>1-\tau_X(q)$. Define $H_q(x_i)$ as in Lemma~\ref{lem:link}; we say that
$H_q(x_i)$ is \emph{defined by} $x_i$. Let $\mathcal{F}:=\{H_q(x_1), \dots, H_q(x_n)\}$. 

Let $\mathcal{F}'$ be a subfamily of $\mathcal{F}$ of $\beta n$ sets.
Let $X'$ be the subset $X$ of points defining the sets in $\mathcal{F}'$. Note that since $q$ has Tukey depth 
greater than $1-\beta$ with respect to $X$, $q$ cannot be separated from $X'$ by a hyperplane $\Pi$. 
Otherwise on the closed halfspace that contains $q$ and that is bounded by $\Pi$, there are at most $(1-\beta) n$ points of $X$. 
Thus, $q \in \conv(X')$. Then by Lemma~\ref{lem:link},
\[\bigcap_{x_i \in X'} H_q(x_i)=\emptyset.\] 
Therefore, by the converse of Theorem~\ref{thm:helly_depth} there exist an $\alpha=c(d,1-\beta)$ and subfamilies $F_1,\dots,F_{d+1}$ of $\mathcal{F}$,
each of at least $\alpha n$ sets, such that for every choice of sets $H_q(x_i') \in F_i$ we have that $\bigcap_{i=1}^{d+1} H_q(x_i') = \emptyset$.
Thus, by Lemma~\ref{lem:link}, for every such choice we have that $q \in \conv(\{x_1',\dots, x_{d+1}'\})$.
The result follows by setting $X_i$ to be be the set of points defining the sets in $F_i$ and letting $\beta$ tend to $1-\tau_X(q)$.
 \end{proof}
 
As with the classical versions, from the Depth Carath\'eodory's or Depth Helly's theorems we can prove 
a ``depth'' version of the Kirchberger's theorem.
 
\begin{theorem}[\textbf{Depth Kirchberger's theorem}]
Let $ 0 \le \beta \le 1$, $d \ge 2$ and $\alpha=c(d,1-\beta)$ (where $c(d,\beta)$ is as in Theorem~\ref{thm:car_depth}). 
 Let $R$ and $B$ be sets of points in $\mathbb{R}^d$ such that $R \cup B$ has $n$ points. Suppose 
 that for every choice of subsets $Y_1,\dots,Y_{d+2}$ of $R \cup B$, each of at least
 $\alpha n$ points, there exists a set $S=\{y_i: y_i \in Y_i\}$ that satisfies the following.
 The set $S \cap R$ can be separated by a hyperplane from $S \cap B$. Then
 there exist subsets $R' \subset R$ and $B' \cap B$ such that $|R' \cup B'| \ge \beta n$, and
 $R'$ can be separated from $B'$ by a hyperplane.
\end{theorem}
\begin{proof}
 We map $R$ and $B$ to $\mathbb{R}^{d+1}$ as follows.
 Let \[\mathcal{R}:=\{(x_1,\dots,x_d,1): (x_1,\dots,x_d) \in R\}\] and \[\mathcal{B}:=\{(-x_1,\dots,-x_d,-1): (x_1,\dots,x_d) \in B\}.\] 
 It can be shown that a subset $\mathcal{S}$ of $\mathcal{R} \cup \mathcal{B}$ can be separated from the origin by a hyperplane
 if and only if its preimage $S$ in $R \cup B$ satisfies the following. The set $S \cap R$ can be separated
 by a hyperplane from $S \cap B$. 
 
 We claim that the Tukey depth of the origin with respect to $\mathcal{R} \cup \mathcal{B}$ is less than
 $1-\beta$. Suppose for a contradiction the Tukey depth of the origin with respect to $\mathcal{R} \cup \mathcal{B}$ 
 is at least $1-\beta$ . By Theorem~\ref{thm:car_depth}  there exist subsets $\mathcal{Y}_1,\dots, \mathcal{Y}_{d+2}$ of $\mathcal{R} \cup \mathcal{B}$
 each of at least $\alpha n$ points that satisfy the following. For every choice of $\mathcal{S}=\{y_i: y_i \in \mathcal{Y}_i\}$, the set
 $\mathcal{S}$ cannot be separated from the origin by a hyperplane.
  Let $Y_1,\dots,Y_{d+2}$ be the preimages of the $\mathcal{Y}_i$'s. Then for every subset $S=\{y_i: y_i \in Y_i\}$, 
 $S \cap R$ cannot be separated from $S \cap B$ by a hyperplane---a contradiction.
 
 Therefore, there exists a subset of $\mathcal{R}\cup \mathcal{B}$ of less than $(1-\beta)n$ points that can 
 be separated from the origin by a hyperplane; the complement $\mathcal{X}$ of this set in $\mathcal{R} \cup \mathcal{B}$ can
 also be separated from the origin and has more than $\beta n$ points. Let $X$ be the preimage of $\mathcal{X}$.
 The result follows by setting $R':=X\cap R$ and $B':=X \cap B$.
\end{proof}

We have used quotes when referring to the depth versions of Helly's and Kirchberger's theorems. We done so
because their relationships with the notion of depth is only in their close relationship to the Depth
Carath\'eodory's theorem. The depth versions of Carath\'eodory's and Helly's theorems seem to be a combination of the colorful
and fractional versions. The hypothesis is that for every subfamily in which every object is assigned
one of $(d+1)$ colors and every color class is large, there exists a colorful $(d+1)$-tuple
that satisfies the property. The conclusion is that a large subfamily satisfies the property.
We conclude the paper by mentioning two other results that are fractional/colorful versions
of their classical counterparts.

Recently, the following colorful fractional Helly's theorem has been 
found by B{\'a}r{\'a}ny, Fodor, Montejano, Oliveros and P{\'o}r.

\begin{theorem}[\textbf{Fractional Colorful Helly's theorem~\cite{fractional_colorful}}]
Let $\mathcal{F}$ be a finite family of convex sets such that each
set is assigned one of $d+1$ colors. Let $\mathcal{F}_1,\dots, \mathcal{F}_{d+1}$
be its color classes. Suppose that for some $\alpha >0$, at least 
$\alpha |\mathcal{F}_1|\cdots|\mathcal{F}_{d+1}|$ of the colorful tuples
have non-empty intersection. Then some $\mathcal{F}_i$ contains a subfamily 
of at least $\frac{\alpha}{d+1} |\mathcal{F}_i|$ sets with a non-empty intersection.
\end{theorem}

The following result was proved by Pach~\cite{homogenous_pach}. It is a fractional/colorful version of the 
Centerpoint theorem; it also bears some resemblance to the Depth Carath\'eodory
theorem.

\begin{theorem}
 There exists a constant $c_d>0$ such that the following holds.
 Let $X_1,\dots,X_{d+1}$ be finite sets of points in general position in $\mathbb{R}^d$.
 Then there exist a point $q$ and subsets $Y_1 \subset X_1,\dots,Y_{d+1} \subset X_{d+1}$ such that the following
 holds. Each $Y_i$ has at least $c_d|X_i|$ points; and for every choice $x_i \in Y_i$, $q$ is contained
 in $\conv(\{x_1,\dots,x_{d+1}\})$.
\end{theorem}

\section*{Acknowledgements}
We thank the anonymous referees whose comments helped us improve our paper
significantly.

\bibliographystyle{plain}
\bibliography{car}

\begin{thebibliography}{10}

\bibitem{afshani}
Peyman Afshani.
\newblock {\em {O}n {G}eometric {R}ange {S}earching, {A}pproximate {C}ounting
  and {D}epth {P}roblems}.
\newblock PhD thesis, University of Waterloo, 2008.

\bibitem{survey_helly_1}
Nina {Amenta}, J{\'e}sus~A. {De Loera}, and Pablo {Sober{\'o}n}.
\newblock {Helly's Theorem: New Variations and Applications}.
\newblock {\em Contemporary Mathematics}.
\newblock To appear.

\bibitem{very_colorful}
Jorge~L. Arocha, Imre B{\'a}r{\'a}ny, Javier Bracho, Ruy Fabila, and Luis
  Montejano.
\newblock Very colorful theorems.
\newblock {\em Discrete Comput. Geom.}, 42(2):142--154, 2009.

\bibitem{colorful_caratheodory}
Imre B{\'a}r{\'a}ny.
\newblock A generalization of {C}arath\'eodory's theorem.
\newblock {\em Discrete Math.}, 40(2-3):141--152, 1982.

\bibitem{fractional_colorful}
Imre B{\'a}r{\'a}ny, Ferenc Fodor, Luis Montejano, Deborah Oliveros, and Attila
  P{\'o}r.
\newblock Colourful and fractional {$(p,q)$}-theorems.
\newblock {\em Discrete Comput. Geom.}, 51(3):628--642, 2014.

\bibitem{same_type}
Imre B{\'a}r{\'a}ny and Pavel Valtr.
\newblock A positive fraction {E}rd{\H o}s-{S}zekeres theorem.
\newblock {\em Discrete Comput. Geom.}, 19(3, Special Issue):335--342, 1998.
\newblock Dedicated to the memory of Paul Erd{\H{o}}s.

\bibitem{boros}
Endre Boros and Zoltan F{\"u}redi.
\newblock The number of triangles covering the center of an {$n$}-set.
\newblock {\em Geom. Dedicata}, 17(1):69--77, 1984.

\bibitem{stabbing}
Boris Bukh, Ji{\v{r}}{\'{\i}} Matou{\v{s}}ek, and Gabriel Nivasch.
\newblock Stabbing simplices by points and flats.
\newblock {\em Discrete Comput. Geom.}, 43(2):321--338, 2010.

\bibitem{caratheodory}
Constantin Carath{\'e}odory.
\newblock \"{U}ber den {V}ariabilit\"atsbereich der {K}oeffizienten von
  {P}otenzreihen, die gegebene {W}erte nicht annehmen.
\newblock {\em Math. Ann.}, 64(1):95--115, 1907.

\bibitem{convex_book}
Harold~G. Eggleston.
\newblock {\em Convexity}.
\newblock Cambridge Tracts in Mathematics and Mathematical Physics, No. 47.
  Cambridge University Press, New York, 1958.

\bibitem{order_type}
Jacob~E. Goodman and Richard Pollack.
\newblock Multidimensional sorting.
\newblock {\em SIAM J. Comput.}, 12(3):484--507, 1983.

\bibitem{gromov}
Mikhail Gromov.
\newblock Singularities, expanders and topology of maps. {P}art 2: {F}rom
  combinatorics to topology via algebraic isoperimetry.
\newblock {\em Geom. Funct. Anal.}, 20(2):416--526, 2010.

\bibitem{helly}
Eduard Helly.
\newblock \"{U}ber mengen konvexer k\"{o}rper mit gemeinschaftlichen punkte.
\newblock {\em Jahresbericht der Deutschen Mathematiker-Vereinigung},
  32:175--176, 1923.

\bibitem{survey_helly_2}
Andreas Holmsen and Rephael Wenger.
\newblock Helly-type theorems and geometric transversals.
\newblock In {\em Handbook of discrete and computational geometry}, CRC Press
  Ser. Discrete Math. Appl., pages 63--82. CRC, Boca Raton, FL, second edition,
  1997.

\bibitem{points_surrounding}
Andreas~F. Holmsen, J{\'a}nos Pach, and Helge Tverberg.
\newblock Points surrounding the origin.
\newblock {\em Combinatorica}, 28(6):633--644, 2008.

\bibitem{karasev}
Roman Karasev.
\newblock A simpler proof of the
  {B}oros-{F}\"uredi-{B}\'ar\'any-{P}ach-{G}romov theorem.
\newblock {\em Discrete Comput. Geom.}, 47(3):492--495, 2012.

\bibitem{fractional_helly}
Meir Katchalski and Andrew C.~F. Liu.
\newblock A problem of geometry in {${\bf R}^{n}$}.
\newblock {\em Proc. Amer. Math. Soc.}, 75(2):284--288, 1979.

\bibitem{kirchberger}
Paul Kirchberger.
\newblock \"{U}ber {T}chebychefsche {A}nn\"aherungsmethoden.
\newblock {\em Math. Ann.}, 57(4):509--540, 1903.

\bibitem{konig_helly}
D{\'e}nes K{\"o}nig.
\newblock \"{U}ber konvexe {K}\"orper.
\newblock {\em Math. Z.}, 14(1):208--210, 1922.

\bibitem{liu}
Regina~Y. Liu.
\newblock On a notion of data depth based on random simplices.
\newblock {\em Ann. Statist.}, 18(1):405--414, 1990.

\bibitem{data_depth_book}
Regina~Y. Liu, Robert Serfling, and Diane~L. Souvaine, editors.
\newblock {\em Data depth: robust multivariate analysis, computational geometry
  and applications}.
\newblock DIMACS Series in Discrete Mathematics and Theoretical Computer
  Science, 72. American Mathematical Society, Providence, RI, 2006.

\bibitem{lectures}
Ji{\v{r}}{\'{\i}} Matou{\v{s}}ek.
\newblock {\em Lectures on discrete geometry}, volume 212 of {\em Graduate
  Texts in Mathematics}.
\newblock Springer-Verlag, New York, 2002.

\bibitem{on_gromov}
Ji{\v{r}}{\'{\i}} Matou{\v{s}}ek and Uli Wagner.
\newblock On {G}romov's {M}ethod of {S}electing {H}eavily {C}overed {P}oints.
\newblock {\em Discrete Comput. Geom.}, 52(1):1--33, 2014.

\bibitem{further_colorful}
Fr{\'e}d{\'e}ric Meunier and Antoine Deza.
\newblock A further generalization of the colourful {C}arath\'eodory theorem.
\newblock In {\em Discrete geometry and optimization}, volume~69 of {\em Fields
  Inst. Commun.}, pages 179--190. Springer, New York, 2013.

\bibitem{mosler}
Karl Mosler.
\newblock Depth statistics.
\newblock In Claudia Becker, Roland Fried, and Sonja Kuhnt, editors, {\em
  Robustness and {C}omplex {D}ata {S}tructures}, pages 17--34. Springer Berlin
  Heidelberg, 2013.

\bibitem{neumann}
Bernhard~H. Neumann.
\newblock On an invariant of plane regions and mass distributions.
\newblock {\em J. London Math. Soc.}, 20:226--237, 1945.

\bibitem{homogenous_pach}
J{\'a}nos Pach.
\newblock A {T}verberg-type result on multicolored simplices.
\newblock {\em Comput. Geom.}, 10(2):71--76, 1998.

\bibitem{rado}
Richard Rado.
\newblock A theorem on general measure.
\newblock {\em J. London Math. Soc.}, 21:291--300 (1947), 1946.

\bibitem{radon_helly}
Johann Radon.
\newblock Mengen konvexer {K}\"orper, die einen gemeinsamen {P}unkt enthalten.
\newblock {\em Math. Ann.}, 83(1-2):113--115, 1921.

\bibitem{survey_depth}
Eynat Rafalin and Diane~L. Souvaine.
\newblock Computational geometry and statistical depth measures.
\newblock In {\em Theory and applications of recent robust methods}, Stat. Ind.
  Technol., pages 283--295. Birkh\"auser, Basel, 2004.

\bibitem{tukey}
John~W. Tukey.
\newblock Mathematics and the picturing of data.
\newblock In {\em Proceedings of the {I}nternational {C}ongress of
  {M}athematicians ({V}ancouver, {B}. {C}., 1974), {V}ol. 2}, pages 523--531.
  Canad. Math. Congress, Montreal, Que., 1975.

\bibitem{wagner}
Ulrich Wagner.
\newblock {\em On {$k$}-Sets and {A}pplications}.
\newblock PhD thesis, ETH Z{\"u}rich, 2003.

\end{thebibliography}

\end{document}